\documentclass[letterpaper]{article} % DO NOT CHANGE THIS
\usepackage{aaai24}  % DO NOT CHANGE THIS
\usepackage{times}  % DO NOT CHANGE THIS
\usepackage{helvet}  % DO NOT CHANGE THIS
\usepackage{courier}  % DO NOT CHANGE THIS
\usepackage[hyphens]{url}  % DO NOT CHANGE THIS
\usepackage{graphicx} % DO NOT CHANGE THIS
\usepackage{natbib}  % DO NOT CHANGE THIS AND DO NOT ADD ANY OPTIONS TO IT
\usepackage{caption} % DO NOT CHANGE THIS AND DO NOT ADD ANY OPTIONS TO IT
\usepackage{algorithm}
\usepackage{algorithmic}

\usepackage{newfloat}
\usepackage{listings}
\DeclareCaptionStyle{ruled}{labelfont=normalfont,labelsep=colon,strut=off} % DO NOT CHANGE THIS
\floatstyle{ruled}
\newfloat{listing}{tb}{lst}{}
\floatname{listing}{Listing}
\usepackage{soul}  
\usepackage[utf8]{inputenc} 
\usepackage{amsmath}
\usepackage{amssymb}
\usepackage{amsthm}
\usepackage{booktabs}  
\usepackage[usenames,dvipsnames]{xcolor} 
\usepackage{macros} 
\usepackage{tikz}
\usetikzlibrary{arrows,positioning} 
\usepackage{multirow}

\newtheorem{example}{Example}
\newtheorem{theorem}{Theorem}
\usepackage{newfloat}
\usepackage{listings}
\floatstyle{ruled}
\newfloat{listing}{tb}{lst}{}
\floatname{listing}{Listing}
\title{Natural Strategic Ability in Stochastic Multi-Agent  Systems \footnote{This is an extended version of the same title paper that appeared at AAAI 2024, containing a technical appendix with proof details. In this version, we also correct the proof of Theorem \ref{prop:mcheck-det-patl}. We would like to thank Wojtek Jamroga for spotting the problem
and alerting us. }}

\author {
    Raphaël Berthon\textsuperscript{\rm 1}, 
    Joost-Pieter Katoen\textsuperscript{\rm 1},
    Munyque Mittelmann\textsuperscript{\rm 2},
    Aniello Murano\textsuperscript{\rm 2}
}
\affiliations {
    \textsuperscript{\rm 1} RWTH Aachen University, Germany\\
    \textsuperscript{\rm 2} University of Naples Federico II, Italy\\
    \{berthon,katoen\}@cs.rwth-aachen.de, \{munyque.mittelmann, aniello.murano\}@unina.it
}

\begin{document}

\maketitle

\begin{abstract}
Strategies synthesized using formal methods %\PATL or \PATLs,
can be complex and often require infinite memory, 
which does not correspond to the expected behavior when trying to model Multi-Agent Systems (MAS). To capture such behaviors,
natural strategies are a recently proposed framework striking a balance between the ability of agents to strategize with memory and the model-checking complexity, but until now has been restricted to fully deterministic settings. 
For the first time, we consider the probabilistic temporal logics %Probabilistic Alternating-time Temporal Logics 
\PATL and \PATLs under natural strategies (Nat\PATL and Nat\PATLs, resp.).
As main result we show that, in stochastic MAS, Nat\PATL model-checking is $\Delta_2^P$-complete when the active coalition is restricted to deterministic strategies. We also give a \DNEXPTIME complexity result for Nat\PATLs with the same restriction. In the unrestricted case, we give an \EXPSPACE complexity for Nat\PATL and \TEXPSPACE complexity for Nat\PATLs.

% 
%We consider      Nat\PATL and Nat\PATLs, a probabilistic logic for multi-agent  natural strategies, which are an encoding of conditional plans using regular expressions of bounded length.  We extend the previous proposal of natural strategies to enable probabilistic choice of actions.

%The logics considered enable reasoning about the strategic abilities of agents that play  strategies with bounded complexity in the probabilistic setting.  
%Our main contribution are the novel complexity results for  \PATL and \PATLs with  memoryless and bounded recall agents when considering both probabilistic and deterministic natural strategies. 
\end{abstract}

\section{Introduction}

In the last decade, much attention has been devoted to the verification of \textit{Multi-Agent Systems} (MAS).   
One of the most important early developments was the Alternating-time Temporal Logics \ATL and \ATLs \citep{AlurHK02}. %
Since its initial proposal, \ATL\ has been extended in various directions, considering, for instance, strategy contexts~\citep{DBLP:journals/iandc/LaroussinieM15} or adding imperfect information  and epistemic operators
~\citep{jamroga2011comparing}. 
Strategy Logic (\SL) \citep{ChatterjeeHP10,MMPV14} extends \ATL to treat  strategies as first-order variables. 
The probabilistic logics \PATL, \PATLs~\cite{chen2007probabilistic}, Stochastic Game Logic~\cite{DBLP:journals/acta/BaierBGK12}, and \PSL~\cite{aminof2019probabilistic}  enhances \ATL, \ATLs, \ATL with strategy contexts, and  \SL, resp., to the probabilistic setting. 
Those logics allow us to express that a coalition can enforce that the probability of satisfying their goal meets 
a specified constraint.

The importance of the aforementioned logics lies in the \textit{uncertainty} often faced by MAS, due to the occurrence of randomization, such as natural events and the behavior of their components (i.e., the agents). While those aspects cannot be known with certainty, they can be measured based on experiments or past observations. Examples include, among others, the affluence of users interacting with the system, unknown preference of its agents modeled with probabilistic distributions, and errors of its sensorial components.
All the aforementioned logics also have downsides, either complexity-wise or memory-wise. \PSL is undecidable, and is still \TEXPSPACE when restricted to  memoryless strategies. 
\PATL model checking is in \NPInter but requires infinite-memory strategies. Stochastic game logic is \PSPACE with memoryless deterministic strategies, and \EXPSPACE with memoryless probabilistic strategies. These last two results are of interest, but the memoryless assumption is quite restrictive.

\emph{Natural strategies}, first defined in~\cite{natStrategy}, are lists of condition-action pairs with a bounded memory representation. 
This definition contrasts with combinatorial  strategies (i.e., functions from histories to actions), considered typically in the semantics of logics for MAS, including \ATL\ and \ATLs. 
The motivation for natural strategies, as argued in \cite{natStrategy}, is that  combinatorial strategies are not   realistic in the context of human behavior, because of the difficulty to execute and design complex plans.
In particular, systems 
that are difficult to use are often ignored by the users, even if they respect design specifications such as security constraints. Artificial agents with limited memory or computational power cannot use combinatorial strategies either. On the other end of the spectrum, memoryless strategies that depend only on the current state cannot provide adequate solutions to many planning problems.

Natural strategies encompass both bounded memory and specifications of agents with “simple” strategies, by allowing agents to use some past observations without requiring infinite memory. They aim at capturing the intuitive approach a human would use when describing strategies. As a result, these strategies are easier to explain using natural language. They also intrinsically feature imperfect information, since they reason about the sequence of propositional variables observed in previous states, instead of the states themselves. Although the systems with whom these agents interact may be stochastic,  the study of natural strategies has been until now restricted to fully deterministic settings. \emph{For the first time, we consider  \PATL\ and \PATLs under natural strategies and  investigate their model checking problem for stochastic MAS.} Remarkably, the logics we consider can also be seen as an extension of POMDPS to a setting with multiple agents with bounded memory strategies~\cite{DBLP:conf/aaai/ChatterjeeCD16}. %relation to properties expressed in the variants of \PATL\ and \PATLs with natural strategies.

\paragraph{Contribution.}
In this paper, we propose variants of the probabilistic logics \PATL\ and \PATLs with natural strategies  (denoted Nat\PATL\ and Nat\PATLs, resp.) and study their complexity for model checking. We present complexity results for deterministic and, for the first time, \textit{probabilistic natural strategies}. With respect to the agents' memory, we investigate both the memoryless and bounded recall settings~\footnote{As usual, we denote no recall with r and recall with R. }. Table~\ref{tab:mcheck-complexity} summarizes the results of this paper. The main advantage of the logics proposed is that they enable to express and verify the strategic abilities  of stochastic MAS in which agents have limited memory and/or computational power, with a reasonably good model checking complexity. In particular, the model checking of Nat\PATL[R] is $\Delta_2^P$-complete for deterministic natural strategies, and in \EXPSPACE for probabilistic natural strategies. %We illustrate the  usefulness of this setting with meaningful examples given throughout the paper.

\begin{table}[] 
\centering
\begin{tabular}{llll}
\hline
 &  & Det.$\sim$Strategies & Prob.$\sim$Strategies \\ \hline
 & Nat\PATL[r] & $\Delta_2^P$-complete & \EXPSPACE \\
 & Nat\PATLs[r] & \DNEXPTIME & \TEXPSPACE \\
 & Nat\PATL[R]& $\Delta_2^P$-complete & \EXPSPACE \\  
 & Nat\PATLs[R] & \DNEXPTIME & \TEXPSPACE \\  \hline 
\end{tabular}
\caption{Summary of model checking complexity problems for Nat\PATL and Nat\PATLs with stochastic  MAS. 
}  
\label{tab:mcheck-complexity}
\vspace{-0.5cm}
\end{table}

\paragraph{Outline.} We start the paper by presenting related work and preliminary definitions. Then, we introduce behavioral natural strategies and the logics Nat\PATL and Nat\PATLs. Next, we discuss motivating examples. We proceed by presenting technical results on the model checking complexity and expressivity. Finally, we conclude the paper.

 \section{Related Work}

Several works consider the verification of  stochastic MAS with specifications given in probabilistic logics. In particular,   
\citet{huang2013logic} study an ATL-like logic for stochastic MAS when agents play deterministic strategies and have probabilistic knowledge. %
The model checking problem has been studied for Probabilistic Alternating-Time $\mu$-Calculus \citep{song2019probabilistic}.
\citet{Huang2012}  consider the logic Probabilistic \ATLs\ (\PATLs) under incomplete information and synchronous perfect recall. \PATL was  also  considered under imperfect information and memoryless strategies \cite{jamroga2023patlii}, and with 
 accumulated costs/rewards \cite{chen2013automatic}.

Also in the context of MAS, probabilistic logics  were  used for the verification of unbounded parameterized systems \citep{lomuscio2020parameterised}, resource-bounded systems \citep{nguyen2019probabilistic}, and under assumptions over opponents' strategies \citep{DBLP:journals/fuin/BullingJ09}.

Our work is also related to the research on representation of strategies with limited memory.
This includes  the  representation of finite-memory strategies by input/output automata~\cite{Vester13}, decision trees~\cite{DBLP:conf/cav/BrazdilCCFK15}, \ATL with bounded memory~\cite{Agotnes09bounded}, as well as the use of bounded memory as an approximation of perfect recall~\cite{BLM18}. More recently, \citet{DEUSER2020103399}
 represented  strategies as Mealy
machines and studied how bounded recall affects the agents’ abilities to execute plans. % 

Natural strategies have first been studied in~\cite{natStrategy} on multiple deterministic settings: finding winning strategies in concurrent games with \LTL specifications, deciding if a set of strategies defines a Nash equilibrium, and model checking \ATL. This last use of natural strategies has later been extended %
to \ATL with imperfect information 
 \cite{natStrategyII} and \SL\  \cite{DBLP:conf/atal/BelardinelliJMM22}.

The study of partially observable MDPs (POMDPs) also considers a variety of strategy representations, as discussed in~\cite{DBLP:journals/toct/VlassisLB12}. When allowing infinite-memory strategies, finding an almost-sure winning strategy with a B\"uchi or reachability objective requires exponential time on POMDPs, while finding strategies for almost-sure parity objectives~\cite{baier2008decision,chatterjee2010qualitative} and for maximizing a reachability objective~\cite{DBLP:journals/ai/MadaniHC03} is undecidable. 
However, when resticting the memory of the strategies to some fixed bound~\cite{DBLP:conf/nips/PajarinenP11,DBLP:conf/uai/Junges0WQWK018}, the complexity of threshold reachability becomes \ETR-complete (the existential theory of the reals) with probabilistic strategies and \NP-complete with deterministic strategies~\cite{DBLP:phd/dnb/Junges20}. The complexity of almost-sure reachability with bounded memory probabilistic strategies is also \NP-complete~\cite{DBLP:conf/aaai/ChatterjeeCD16}.

\section{Preliminaries}
\label{sec:preliminars}

In this paper, we fix finite non-empty sets of agents $\Ag$, actions $\Act$,
and atomic propositions $\APf$. 
We write $\profile{\act}$ for a tuple of actions $(\act_\ag)_{\ag\in\Ag}$, one for each agent, and such tuples are called \emph{action profiles}.
Given an action profile $\profile{\act}$ and $\coalition\subseteq\Ag$, we let $\act_\coalition$ be the components of agents in  $\coalition$, and $\profile{\act}_{-\coalition}$ is $(\act_\agb)_{\agb\not \in \coalition}$. Similarly, we let $\Ag_{-\coalition}=\Ag\setminus\coalition$.

\halfline
\head{Distributions. } Let $X$ be a finite non-empty set. A \emph{(probability) distribution} over $X$ is a function $\distribution:X \to [0,1]$ such that $\sum_{x \in X} \distribution(x) = 1$. Let $\Dist(X)$ be the set of distributions over $X$. We write $x \in \distribution$ for $\distribution(x) > 0$. 
If $\distribution(x) = 1$ for some element $x \in X$, then $\distribution$ is a \emph{point (a.k.a. Dirac) distribution}. 
If, for $i\in I$, $\distribution_i$ is a distribution over $X_i$, then, writing $X = \prod_{i\in I} X_i$, the \emph{product distribution} of the $\distribution_i$ is the distribution $\distribution:X \to [0,1]$ defined by $\distribution(x) = \prod_{i\in I} \distribution_i(x_i)$.

\halfline
\head{Markov Chains. }
  A \emph{Markov chain} $M$ is a tuple $(\setpos,p)$ where $\setpos$ is a countable non-empty set of states and $p \in \Dist(\setpos \times \setpos)$ is a distribution. For $s,t\in\setpos$, the values $p(s,t)$ are called \emph{transition probabilities} of $M$.
  A \emph{path} is an infinite sequence of states. %

\halfline
\head{Concurrent Game Structures. } 
  A \emph{stochastic concurrent game structure} (or simply \emph{CGS})
  $\System$ is a tuple 
  $ ( \setpos, \legal, 
  \trans, \val)$ where 
  (i) $\setpos$ is a finite non-empty set of \emph{states};
   (ii) $\legal: \setpos \times \Ag \to 2^\Act\setminus\{\emptyset\}$ is a \emph{legality function} defining the available actions for each agent in each state, we write $\profile{\legal(\pos)}$ for the tuple $(\legal(\pos, \ag))_{\ag\in\Ag}$; 
  (iii)   for each state $\pos \in \setpos$ and each  move $\mov \in \profile{\legal(\pos)}$, the \emph{stochastic transition function} $\trans$ gives the (conditional) probability $\trans(\pos, \mov)(s')$ of a transition from state $\pos$ for all $\pos' \in \setpos$ if each player $\ag \in \Ag$ plays the action $\mova$, and remark that $\trans(\pos, \mov)\in\Dist(\setpos)$;
  (iv) $\val:\setpos \to 2^{\APf}$ is a \emph{labelling function}.

For each state $\pos \in \setpos$ and joint action $\mov \in \prod_{\ag \in \Ag} \legal(\pos,\ag)$, we assume 
  that there is a state $\pos'\in\setpos$ such that $\trans(\pos, \mov)(\pos')$ is non-zero, that is, every state has a successive state from a legal move, formally $\mov\in \legal(\pos,\ag)$.

\begin{example}[Secure voting \footnote{Our running example on secure voting is adapted from the case study from  \cite{jamroga2020natural,jamroga2022measure}. }]
\label{ex:voting}

Assume a voting system with two types of agents: voters and  coercers, represented by the disjoint sets $V \subset \Ag $ and $C\subset \Ag$, resp. We consider a finite set of %candidates,
receipts, and signatures. 
The actions of the voters  are  $scanBallot$, $enterVote$, $cnlVote$, $\textit{conf}$, $checkSig_s$, $checkrec_r$, $shred_r$, and $noop$, which  represent that the agent is  scanning the ballot,  entering their vote, canceling it, \textit{conf}irming it, checking its signature $s$,  checking the receipt $r$, shredding the receipt $r$, and doing nothing, resp. On its turn, the coercer can perform the actions $coerce_{v}$, $request_v$, $punish_v$, and $noop$, representing that she is coercing the voter $v$,
requesting $v$ to vote, punishing $v$, and doing nothing, resp.  

The CGS has propositions denoting the state of the voting system. Specifically, they describe whether the voter $v$ was coerced ($coerced_v$), punished ($punished_v$), requested to vote ($requested_v$), has a ballot available ($hasBallot_v$), scanned the ballot ($scanned_v$), entered the vote which has the signature $s$  ($entVote_{v,s}$), and has already voted ($vot_v$). For a signature $s$, the proposition $sigOk_s$ 
denotes whether the signature $s$ was checked and corresponds to the one in the system, while the proposition $sigFail_s$ denotes that it was checked but did not correspond. For a receipt $r$, the propositions $rec_{v,r}$ %,  $receiptFail_r$ 
and $shreded_r$ denotes whether   $r$ associated with voter $v$  
and whether $r$ was destroyed (and it's no longer visible), resp.  

Actions performed by the agents may fail and may not change the state of the system as intended by them. 
For instance,  the coercer may not succeed  (attempting) to coerce a voter with  the action $coerce_{v}$ (and thus, $coerced_v$ may not be true in the next state). Similarly, a voter's request to shred her receipt  may fail, and the information on the receipt be still visible. The probability of an action failing is described by the CGS stochastic transition function.  
\end{example} 

\halfline
\head{Plays. } 
A \emph{play} or path in a CGS $\System$ is an infinite sequence $\iplay=\pos_0 \pos_1 \cdots$ of states
such that there exists a sequence $\mov_0 \mov_1 \cdots$ of joint-actions such that $\mov_i \in \legal(\pos_{i})$ and  $\pos_{i+1} \in \trans(\pos_i,\mov_i)$ (\ie, $\trans(\pos_i,\mov_i)(\pos_{i+1} )>0$) for every $i \geq 0$.
We write $\iplay_i$ for $\pos_i$, 
$\iplay_{\geq i}$ for the suffix of
$\iplay$ starting at position $i$. 
Finite paths are called \emph{histories}, and the set of all histories is denoted $\History$. We write $\last(\history)$ for the last state of a history $\history$ and $len(\history)$ for the size of $\history$.

\section{Behavioral Natural Strategies}

In this section, we define behavioral\footnote{Behavioral strategies define the probability of taking
an \textit{action} in a state. This is different from mixed strategies, which define the probability of taking a  \textit{strategy} in a game. The relation of behavioral and mixed strategies is discussed in \cite{kaneko1995behavior}.} natural strategies over \CGS, based on the definition in \cite{natStrategy}, and use them to provide the semantics of Nat\ATLs. 
Natural strategies are conditional plans, represented through an ordered list of condition-action rules. The intuition is that the first rule whose condition holds in the history of the game is selected, and the corresponding action is executed. 
The conditions are regular expressions over \textit{Boolean formulas} over $\APf$, denoted $Bool(\APf)$ and given by the following BNF grammar: 
\[\phi \coloncolonequals p \mid \phi \lor \phi \mid \neg \phi 
\text{\hspace{2em} where  $p \in \APf$.}\] 
   
Given a state $\pos \in \setpos$ and a formula $\phi \in Bool(\APf)$, we inductively define the satisfaction value of $\phi$ in $\pos$, denoted $\pos \models \phi$, as follows: 
\begingroup  \allowdisplaybreaks
  \begin{align*}
    \pos &\models p &\text{ iff }& p\in\val(\pos) \\
    \pos &\models \phi_1 \lor \phi_2 &\text{ iff }&  \pos \models \phi_1 \text{ or } \pos \models \phi_2\\
    \pos &\models \neg \phi &\text{ iff }& \text{not } \pos \models \phi  
  \end{align*}
  \endgroup

Let $\setregular(Bool(\APf))$ be the set of regular expressions over the conditions $Bool(\APf)$, defined with the constructors $\concat, \ndchoice, \iteration$ representing concatenation, nondeterministic choice, and Kleene iteration, respectively. 
Given a regular expression $\regular$ and the language $\Language(\regular)$ of finite words generated by $\regular$, a history $\history$ is \textit{consistent} with $\regular$ iff there exists a word $b \in \Language(\regular)$ such that $|\history| = |b|$ and $\history[i] \models b[i]$, for all $0 \leq i \leq |\history|$. 
Intuitively, a history $\history$ is consistent with a regular expression $\regular$ if the $i$-th  epistemic condition in $\regular$ holds in the $i$-th state of $\history$ (for any position $i$ in $\history$). 

A \textit{behavioral natural strategy $\strat$ with recall} for an agent $\ag \in \Ag$ is a sequence of pairs $(\regular, \Dist(\Act))$, where $\regular \in \setregular(Bool(\APf))$ is a regular expression representing recall, and $\distribution(\Act)$ is a distribution over the actions with $\distribution(\act)\neq0$ if $\act$ is available for $\ag$ in $\last(\history)$ (i.e., for $\act \in \legal(\ag, \last(\history))$), for all histories $\history$ %$ \in \sethistory$ 
consistent with $\regular$. The last pair in the sequence is required to be  $(\top\iteration, \distribution(\Act))$, with $\distribution(\act) = 1$ for some $\act \in \legal(\pos,\ag)$ and every $\pos \in \setpos$.   
A \textit{behavioral memoryless natural strategy} is a behavioral natural strategy without recall: each condition is a Boolean formula (i.e., all regular expressions have length 1). 
A strategy $\sigma$ is deterministic if for all pairs $(\regular, \distribution)$, we have $|\{\act\in \Act\ |\ \distribution(\act)\neq 0\}| = 1$. For readability of the examples, given a pair $(\regular, \distribution)$, we write  $(\regular, \act)$  if $\distribution(\act)=1$ for some action $\act \in \Act$.

\begin{example}[Secure voting, continued]
Recall the voting system introduced in Example 
\ref{ex:voting}. 
The following is a deterministic memoryless  natural strategy for the voter $v$:

\begin{enumerate}
%\small{
\item  \label{p:v1} $(hasBallot_v \land \neg scanned_v, scanBallot)$
    \item \label{p:v2} $(\neg vot_v\land scanned_v, enterVote)$
    \item \label{p:v3} $(\neg vot_v\land entVote_{v,s} \land \neg (sigOk_s \lor sigFail_s), checkSig_s)$, for each signature $s$
    \item \label{p:v4} $(\neg vot_v\land entVote_{v,s} \land sigFail_s, \allowbreak 
    cnlVote)$, for each $s$
    \item \label{p:v5} $(\neg vot_v\land entVote_{v,s}\land  sigOk_s, \textit{\textit{conf}})$, for each $s$
    \item \label{p:v6}  $(vot_v \land rec_{v,r} \land \neg shreded_r, shred_r)$, for each receipt $r$
    \item \label{p:v7} $(\top, noop)$ 
%}
\end{enumerate}

This strategy specifies that the agent first scans the ballot in case there is one, and it was not scanned (Pair \ref{p:v1}). Otherwise, if the agent has not voted yet and has scanned, she enters her vote (Pair \ref{p:v2}). If the agent did not vote, entered the vote and did not check the signature, she checks it (\ref{p:v3}). 
When the signature is checked, the agent chooses to cancel or \textit{\textit{conf}}irm the vote, depending on whether the verification has failed or succeeded (Pairs \ref{p:v4} and \ref{p:v5}).  If the agent has voted and there is an unshredded visible receipt, the agent requests it to be shredded (Pair \ref{p:v6}). Finally, if none of the previous conditions apply, the agent does not do any action (Pair \ref{p:v7}).

A behavioral natural strategy with recall for a coercer is:

  \begin{enumerate}
      \item \label{p:c1} $  \big(\top\iteration\concat \bigwedge_{v\in V} \neg coerced_v, \distribution_V\big)$, where $\distribution_V$  is a probability distribution over the actions for coercing the voters, with  $\sum_{v \in V}\distribution_V(coerce_v)= 1$
      \item \label{p:c2} $(\top\iteration\concat coerced_v \land \neg requested_v, request_v)$, for $v \in V$      
      \item \label{p:c3} $(\top\iteration \concat \neg requested_v\iteration \concat (requested_v \land \neg vot_v)\iteration \land \neg  punished_v,  punish_v)$, for each $v \in V$  
      \item \label{p:c4} $(\top \iteration \concat \neg coerced_v \land \neg coerced_{v'}, \distribution_{v,v'})$,  where $\distribution_{v,v'}$  is a probability distribution over the actions for coercing the voters $v$ and $v'$ , with  $\distribution_{v,v'}(coerce_v)+\distribution_{v,v'}(coerce_{v'})= 1$, for each pair $(v, v')$ of distinct voters in $V$
      \item \label{p:c5} $(\top\iteration, noop)$
  \end{enumerate}

  This behavioral strategy considers first the situation in which no voter was already coerced, and the agent chooses the action to coerce one of them randomly (Pair \ref{p:c1}). Next condition-action pair in the strategy says that if a voter was coerced, but her vote was not requested, the agent requests her vote (Pair \ref{p:c2}). If the voter was requested  (at least once in the history), but (continually) did not vote and was not punished, the agent  punishes her (Pair \ref{p:c3}). Next pair says that if there are two distinct non-coerced voters, the agent randomly chooses one to coerce (Pair \ref{p:c4}). If none of those conditions apply, no operation is performed
  (Pair \ref{p:c5}).

\end{example} 

 Throughout this paper, let $\setting \in \{r, R\}$ denote whether we consider memoryless or recall strategies respectively. 
Let $\setstrata^{\setting}$ be the set of behavioral  
natural strategies for agent $\ag$ and  $\setstrat^{\setting} = \cup_{\ag \in \Ag} \setstrata^{\setting}$. %

Let $\size(\strat_\ag)$ denote the number of guarded actions in $\strat_\ag$, $\cond_\idx(\strat_\ag)$ be the $\idx$-th guarded condition on $\strat_\ag$, $\cond_\idx(\strat_\ag)[j]$ be the $j$-th Boolean formula of the guarded condition $\strat_\ag$, and $\action_\idx(\strat_\ag)$ be the corresponding probability distribution on actions. Finally, $\match(\history, \strat_\ag)$ is the smallest index $\idx \leq \size(\strat_\ag)$ such that for all $0 \leq j \leq |last(\history)|$, $\history[j] \models {\cond_\idx(\strat_\ag)[j]}$\footnote{Note that, as in \cite{natStrategy}, we consider the case in which the condition has the same length of the history. There is also the case in which the condition is shorter than the history. This is due to the usage of the Kleene iteration operator. In the latter case, we need to check a finite number of times the same Boolean formula in different states of the history.} and $\action_\idx(\strat_\ag) \in \legal(\ag, \lasth(\history))$. In other words, $\match(\history, \strat_\ag)$ matches the state $\lasth(\history)$ with the first condition in $\strat_\ag$ that holds in $\history$, and action available in $\lasth(\history)$. 

Given a natural strategy $\strat_\ag$ for agent $\ag$ and
a history $\history$, we denote by $\strat_\ag(\history)$ the probability distribution of actions assigned by $\strat_\ag$ in the last state of $\history$, i.e., $\strat_\ag(\history) = act_{match(\history, \strat_\ag)(\strat_\ag)}$.

\paragraph{Complexity of Natural Strategies.}
The complexity $c(\strat)$ of strategy $\strat$ is the total size of its representation and is denoted as follows:   $c(\strat) \egdef \sum_{(\regular, \distribution) \in \strat}|\regular|$,  where $|\regular|$ is the number of symbols in $\regular$, except parentheses.   %

\paragraph{Relation to Game Theory.}
From a game-theoretic point of view, natural strategies can be encoded as finite-memory strategies using finite state machines (e.g., finite-state transducers) that only read the propositional variables holding in a state (akin to imperfect information). Natural strategies and finite-state machines are fundamentally different in their encoding, in particular, the finite state machine may be exponential in the size of the natural strategy it is associated with, as proved in Theorem 12 of~\cite{natStrategy}. 

\section{\PATL and \PATLs with Natural Strategies}
Next, we introduce the Probabilistic Alternating-Time Temporal Logics \PATLs and \PATL with Natural Strategies (denoted, Nat\PATLs and Nat\PATL, resp). 

\begin{definition}\label{def:ATLsF-syntax}
The syntax of Nat\PATLs  is defined by the grammar: 
\begin{align*}
	\varphi  ::= p \mid  {\varphi \lor  \varphi} \mid \neg \varphi \mid \X \varphi \mid \varphi \until \varphi \mid \coop{\coalition}^{\bowtie d}_{k} \varphi
\end{align*}
where $p \in \APf$, $k \in \mathbb{N}$, $\coalition \subseteq \Ag$, $d$ is a rational constant in
$[0, 1]$, and $\bowtie \in  
\{\leq, <, >, \geq\}$. 
\end{definition}

The intuitive reading of the operators is as follows: ``next'' $\X$ and ``until'' $\U$ are the standard temporal operators. $\coop{\coalition}^{\bowtie d}_{k}\varphi$ asserts that there exists a strategy with complexity at most $k$ for the coalition $\coalition$ to collaboratively enforce $\varphi$  with a probability in relation $\bowtie$ with constant $d$.
We  make use of the usual syntactic sugar ${\F \varphi \colonequals \top \U \varphi}$ and ${\G \varphi \colonequals \neg \F \neg \varphi}$ for temporal operators.

A Nat\PATLs formula of the form $\coop{\coalition}^{\bowtie d} \varphi$ %
is also called state formula.   
An important syntactic restriction of Nat\PATLs, namely  Nat\PATL, is defined as follows.  

\begin{definition}[Nat\PATL\ syntax] \label{def:ATLF-syntax}
	The syntax of Nat\PATL  is defined by the grammar
	\begin{align*}   	
		\varphi ::= p \mid  \varphi \lor \varphi \mid \neg \varphi \mid \coop{\coalition}^{\bowtie d}_{k} (\X \varphi) \mid \coop{\coalition}^{\bowtie d}_{k}(\varphi \until\varphi)	
  \end{align*}
where $p$, $k$, $\coalition$, $d$, and $\bowtie$ are as above.  
\end{definition}

Before presenting the semantics, we show how to define the probability space on outcomes.

 \paragraph{Probability Space on Outcomes. } An \emph{outcome} of a
 strategy profile $\profile\strat = (\strat_\ag)_{\ag \in \Ag}$ and a state $\pos$ 
 is a play $\iplay$ that starts in state
 $\pos$ and is extended by $\profile\strat$, formally $\iplay_{0} =
 \pos$, and for every $k \geq len(h)$ there exists $\mov_k \in
 (\strat_\ag(\iplay_{\leq k}))_{\ag \in \Ag}$ such that $\iplay_{k+1} \in
 \trans(\iplay_k,\mov_k)$. 
 The set of outcomes of a strategy profile  $\profile\strat$ and state $\history$  is denoted $Out(\profile\strat,\pos)$.   
 A given %
 CGS 
 $\System$, strategy profile $\profile{\strat}$, and state 
 $\pos$ induce an infinite-state Markov chain
 $M_{\profile{\strat},\pos}$ whose states are the histories in
 $Out(\profile{\strat},\pos)$. 
   and whose transition probabilities
are defined as  $p(\history,\history\pos')=\sum_{\mov\in\Act^\Ag}
 \profile{\strat}(\history)(\mov) \times
 \trans(\last(\history),\mov)(\pos')$.  
 The Markov chain
 $M_{\profile{\strat},\pos}$ induces a canonical probability space on
 its set of infinite paths~\citep{kemeny1976stochastic}, which can be identified with the set of plays in  $Out(\profile{\strat},\pos)$ and the corresponding measure is denoted $out(\profile{\strat},\pos)$.
 ~\footnote{This is a classic construction, see for instance
 ~\citep{clarke2018model,berthon2020alternating}.
 }

  Given a coalition strategy $\profile{\strat_\coalition} \in \prod_{\ag \in \coalition} \setstrata^{\setting}$, the set of possible outcomes of $\profile{\strat_\coalition}$ from a state $\pos \in \setpos$ to be the set  $out_\coalition(\profile{\strat_\coalition},\pos) = \{out((\profile{\strat_\coalition},\profile{\strat_{-\coalition}}),\pos) : \profile{\strat_{-\coalition}} \in \prod_{\ag \in \Ag_{-\coalition}} \setstrata^{\setting} \}$ of probability
measures that the players in $\coalition$ enforce when they
follow the strategy $\profile{\strat_\coalition}$, namely, for each $\ag \in \Ag$,
player $\ag$ follows strategy $\strat_\ag$ in $\profile{\strat_\coalition}$. We use $\mu^{\profile{\strat_\coalition}}_\pos$ to range over the measures in $out_\coalition(\profile{\strat_\coalition},\pos)$ as follows:

\begin{definition}[Nat\PATL and Nat\PATLs semantics]
Given a setting $\setting \in \{r,R\}$, Nat\PATL and Nat\PATLs formulas are interpreted in a stochastic CGS $\System$  
  and a path  $\iplay$,
\begingroup
\allowdisplaybreaks
\begin{align*}
 \System,\iplay &\models_\setting p & \text{ iff } & p \in \val(\iplay_0)\\
 \System,\iplay &\models_\setting \neg \varphi & \text{ iff } & \System,\iplay \not \models_\setting \varphi \\
 \System,\iplay &\models_\setting \varphi_1 \lor \varphi_2 & \text{ iff }&  
 \System,\iplay \models_\setting \varphi_1  \text{ or } \System,\iplay \models_\setting \varphi_2
 \\
\System, \iplay&\models_\setting \coop{\coalition}^{\bowtie d}_k \varphi & \text{ iff }  & 
\exists \profile{\strat_{\coalition}} \in \prod_{\ag \in\coalition}  
 \{\alpha \in \setstrata^{\setting}        : c(\alpha) \leq k\} 
\\ \text{ s.t. }  \forall \mu^{\profile{\strat_\coalition}}_{\iplay_0} \in out_\coalition(\profile{\strat_{\coalition}},\iplay_0)  
\text{, } \mu^{\profile{\strat_\coalition}}_{\iplay_0}(\{\iplay' : \System,\iplay' \models_\setting \varphi\}) \bowtie d\span\span\span
\\ 
 \System,\iplay &\models_\setting \X \varphi & \text{ iff } & \System,\iplay_{\geq 1} \models_\setting \varphi \\
\System, \iplay  & \models_\setting \psi_1 \until \psi_2 & \text{ iff } &  \exists k \geq 0 \text{ s.t. } \System,\iplay_{\geq k} \models_\setting \psi_2 \text{ and } 
\\ & & &
\forall j \in [0,k).\, \,  \System,\iplay_{\geq j}\models_\setting \psi_1
\end{align*} 
\endgroup
\end{definition}

\section{Motivating Examples}\label{sec:ex}

In this section, we present problems that motivate reasoning in stochastic MAS, and we illustrate how Nat\PATLs-formulas can be used to express properties on those systems.

Let us start with an example  of door access control with a random robot. This example illustrates a setting in which it suffices to have deterministic strategies in stochastic CGSs.

\begin{example}[Access control]
    We consider the example illustrated in Figure \ref{fig:robot}.      
    We are given a set $\Ag$ of agents, a set of square tiles, where a non-controlled robot moves randomly either one tile right, left, up, or down at every time step.      
    Between every tile, there is either a wall, a door controlled by some agent with actions $open$ and $close$, or nothing. 
    The robot can cross an empty space, cannot cross a wall, and can only cross a door if the agent controlling has taken action $open$. 
    Given a set of targets represented by atomic propositions $T = \{t_i\in \APf,\ i\in \{1,n\}\}$ labelling some tiles, and related coalitions $\{C_i\subseteq \Ag, i\in  \{1,n\}\}$, we use Nat\PATLs  to state that some coalition $C\subseteq\Ag$ has a strategy with memory $k\in \mathbb{N}$ reaching all targets infinitely often with probability $0.7$, formally: 
\begin{equation} 
\label{eq:ac}\coop{\coalition}^{\geq 0.7}_{k} \G\bigwedge_{t_j\in T, t_j\neq t_i}  \F\;  t_j
\end{equation}

    In the example of Figure~\ref{fig:robot}, where $n=2$, the coalition controls two doors adjacent to the initial state. Even though the structure is probabilistic, memoryless strategies are sufficient. Opening the left door gives the robot a chance to move left, which brings it closer to target $t_0$, but in this center-leftmost square, the agents not in the coalition may open the door leading to the bottom-left square, where they can then trap the robot: the robot only has probability $\frac{1}{2}$ to successfully reach $t_0$, and otherwise may be trapped forever. The other option available to the coalition in the initial state is to close the left door, and open all other doors. The robot will take longer, but has probability $1$ to eventually reach target $t_1$. Thus, we can reach $t_1$ with probability $1$, but $t_0$ with only probability $\frac{1}{2}$, and property~\ref{eq:ac} does not hold.

\begin{figure}[hb]
\centering 
\begin{tikzpicture}
\draw[help lines, step=1] (0,0)grid(3,3);
\path (1.5,1.5) node () {{\Huge$ \Box$}};
\path (1.4,1.6) node () {$\circ$};
\path (1.6,1.6) node () {$\circ$};
\path (1.5,1.4) node () {$\textendash$};ù

\path (1,2.5) node () {$\envir$}; \path (2,2.48) node () {$\system$};
\path (0.5,1.98) node () {$\system$}; \path (1.5,2) node () {$\envir$}; \path (2.5,1.98) node () {$\system$}; 
\path (1,1.48) node () {$\system$}; \path (2,1.5) node () {$\envir$};
\path (0.5,1) node () {$\envir$}; \path (1.5,0.98) node () {$\system$}; \path (2.5,0.98) node () {$\system$}; 
\path (1,0.5) node () {$ $}; \path (2,0.48) node () {$\system$};

\path (0.5,2.5) node () {$t_0$};
\path (1.5,2.5) node () {$t_1$};
\end{tikzpicture}
    \caption{A robot in a maze, where $\system$ and  $\envir$ denote a door respectively controlled by the coalition or the agents not in the coalition. Full lines represent walls. $t_0,t_1$ are two targets.}
    \label{fig:robot}
    \end{figure}
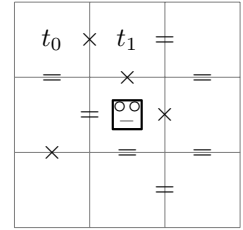
\end{example}

Going back to Example \ref{ex:voting}, we now illustrate how Nat\PATLs and Nat\PATL can be used for the formal security analysis of  voting systems.

\begin{example}[Secure voting, continued]

A requirement usually considered
for electronic voting systems is \textit{voter-verifiability}, which captures the ability of the voter to verify her vote \cite{jamroga2022measure}. In our example, this is represented by the propositions $sigOk_s$ and $sigFail_s$. The Nat\PATL formula
$$\coop{v}_k^{\geq 0.9} \F (sigOk_s \lor sigFail_s)$$ 
says that the voter $v$ has a strategy of size at most $k$
so that, at some point and with probability at least $0.9$, she  obtains either the positive or the negative outcome of verifying the signature $s$.

Another requirement is 
\textit{receipt-freeness}, which expresses that the voters can not gain a receipt to prove that they voted in a certain way. In our example, the propositions $receipt_{v,r}$ and $shreded_r$ represent that a receipt $r$ is associated with the voter $v$ and that the information on it was destroyed. The Nat\PATL formula: 
$$\neg \coop{v}_k^{\geq 0.5} \F \bigvee_{\text{receipt } r}(receipt_{v,r} \land \neg  shreded_r)$$
says that there is no strategy of complexity at most $k$ to ensure with probability at least $0.5$ that, eventually, there will be an unshredded receipt for her.
\end{example}

\section{Model Checking Complexity}\label{sec:mc}

In this section, we look at the complexity of model checking for different versions of Nat\PATL.
\begin{definition}

Given a setting $\setting\in\{r, R\}$, a CGS $\System$, 
state $s \in \setpos$, 
and  formula $\varphi$ in 
Nat\PATL{$\setting$}   (Nat\PATLs{$\setting$}, resp.), 
the model checking problem for
Nat\PATL{$\setting$}
(Nat\PATLs{$\setting$}, resp.) consists in deciding,
whether $\System, s \models_{\setting} \varphi$.
\end{definition}

\begin{definition}[Number of negations]
Given a Nat\PATL[r] (respectively Nat\PATL[R]) formula $\varphi$ without coalition operator, we define subformulas under an even and odd number of negations as follows: $\varphi$ is under an even number of negations, and if $\varphi_1\vee \varphi_2$ or $\X\varphi_1$ or $\varphi_1\until\varphi_2$ is under an even (respectively odd) number of negations, so are $\varphi_1$ and $\varphi_2$. If $\neg\varphi_1$ is under an even (respectively odd) number of negations, then $\varphi_1$ is under an odd (respectively even) number of negations. 
\end{definition}

\begin{theorem}\label{prop:mcheck-det-patl}
Model checking Nat\PATL[r] (respectively Nat\PATL[R]) %
with deterministic natural  strategies for the coalition is in $\Delta_2^P = \Ptime^{\NP}$.
\end{theorem}

\begin{proof}
    We first define the conjugate $\conjugate$ of comparison operators: $\conjugate{\leq}$ is $<$, $\conjugate{<}$ is $\leq$, $\conjugate{>}$ is $\geq$ and $\conjugate{\geq}$ is $>$.    
    We focus on the case of Nat\PATL[R]. Let $\varphi$ be a Nat\PATL[R] formula and $\System = ( \setpos, \legal, \trans, \val)$ be a CGS. 
    We guess a polynomial witness consisting of one deterministic strategy $\sigma[\pos,\varphi',\ag]$ with some bounded complexity $k$ for every $\pos\in \setpos$,  every subformula $\varphi' = \coop{\coalition}^{\bowtie d}_{k}(\varphi'')$ of $\varphi$, and every $\ag\in\Ag$. We now show how to check that formula $\varphi$ holds given such a witness. 
    
    For every formula $\varphi''$ that does not contain any coalition operator and appearing under an even number of negations, it can be decided in polynomial time if $\pos \models \coop{\Ag\backslash \coalition}^{\conjugate{\bowtie} 1-d}_{k}(\neg \varphi'')$ holds when some strategies $\sigma[\pos,\varphi',\ag]$ are fixed for agents $\ag\in\coalition$. Indeed, after fixing such strategies, we obtain an MDP where formula $\neg \varphi''$ can be translated in polynomial time in a polynomial reachability or invariance objective, and checking whether there is probability $\conjugate{\bowtie} 1-d$ for $\neg\varphi''$ to hold on this MDP can be done in polynomial time~\cite{DBLP:books/daglib/0020348}.

    For  formulas $\varphi''$ appearing under an odd number of negations, after fixing some strategies $\sigma[\pos,\varphi',\ag]$  for agents $\ag\not\in\coalition$ it can be decided in polynomial time if $\pos \models \coop{ \coalition}^{\bowtie 1-d}_{k}( \varphi'')$. Indeed, after fixing such strategies, we obtain an MDP where formula $\varphi''$ can be translated in polynomial time in a polynomial reachability or invariance objective, and checking whether there is probability $1-d$ for $\varphi''$ to hold on this MDP can be done in polynomial time~\cite{DBLP:books/daglib/0020348}.

    Going through $\varphi$ in a bottom-up manner, we iteratively replace every subformula $\varphi' = \coop{\coalition}^{\bowtie d}_{k}(\varphi'')$ by a new formula that is true in states $\pos$ iff $\coop{\Ag\backslash \coalition}^{\conjugate{\bowtie} 1-d}_{k}(\neg \varphi'')$ does not hold assuming agents in $\coalition$ follow strategies $\sigma[\pos,\varphi',\ag]$. This can be done in polynomial time, and returns whether $\varphi$ holds or not. 
\end{proof}

\begin{theorem}\label{prop:mcheck-det-patl-h}
Model checking Nat\PATL[r] (respectively Nat\PATL[R]) with deterministic natural  strategies for the coalition is \NP-hard.
\end{theorem}

\begin{definition}[Positive fragment]
The positive fragment of Nat\PATL[r] (respectively Nat\PATL[R]) is the fragment that may use conjunction but no negation. 
\end{definition}

\begin{theorem}\label{prop:mcheck-det-patl}
Model checking the positive fragment of Nat\PATL[r] (respectively Nat\PATL[R]) %
with deterministic natural  strategies for the coalition is in \NP.
\end{theorem}

\begin{proof}
    We first define the conjugate $\conjugate$ of comparison operators: $\conjugate{\leq}$ is $<$, $\conjugate{<}$ is $\leq$, $\conjugate{>}$ is $\geq$ and $\conjugate{\geq}$ is $>$.    
    We focus on the case of Nat\PATL[R]. Let $\varphi$ be a Nat\PATL[R] formula and $\System = ( \setpos, \legal, \trans, \val)$ be a CGS. 
    We guess a polynomial witness consisting of one deterministic strategy $\sigma[\pos,\varphi',\ag]$ with some bounded complexity $k$ for every $\pos\in \setpos$,  every subformula $\varphi' = \coop{\coalition}^{\bowtie d}_{k}(\varphi'')$ of $\varphi$, and every $\ag\in\coalition$. We now show how to check that formula $\varphi$ holds given such a witness. 
    
    For every formula $\varphi''$ that does not contain any coalition operator, it can be decided in polynomial time if $\pos \models \coop{\Ag\backslash \coalition}^{\conjugate{\bowtie} 1-d}_{k}(\neg \varphi'')$ holds when some strategies $\sigma[\pos,\varphi',\ag]$ are fixed for agents $\ag\in\coalition$. Indeed, after fixing such strategies, we obtain an MDP where formula $\neg \varphi''$ can be translated in polynomial time in a polynomial reachability or invariance objective, and checking whether there is probability $\conjugate{\bowtie} 1-d$ for $\neg\varphi''$ to hold on this MDP can be done in polynomial time~\cite{DBLP:books/daglib/0020348}.

    Going through $\varphi$ in a bottom-up manner, we iteratively replace every subformula $\varphi' = \coop{\coalition}^{\bowtie d}_{k}(\varphi'')$ by a new formula that is only true in states $\pos$ where $\coop{\Ag\backslash \coalition}^{\conjugate{\bowtie} 1-d}_{k}(\neg \varphi'')$ does not hold assuming agents in $\coalition$ follow strategies $\sigma[\pos,\varphi',\ag]$. This can be done in polynomial time, and returns whether $\varphi$ holds or not. 
\end{proof}

\begin{theorem}\label{prop:mcheck-det-patl-h}
Model checking Nat\PATL[r] (respectively Nat\PATL[R]) with deterministic natural  strategies for the coalition is \NP-hard.
\end{theorem}

\begin{proof}
    
    We start by showing that Nat\PATL[r] with deterministic natural  strategies for the coalition extends POMDPs with memoryless deterministic strategies and almost-sure reachability objective. Indeed, a POMDP represented as a CGS $\System = ( \setpos, \legal, \trans, \val)$ (two states are indistinguishable if they are labelled by the same propositional variables), a single agent $A$, and a set of target states distinguished by some propositional variable $t$ that holds only in these states, there exists a strategy almost surely reaching $t$ from an initial state $s\in\setpos$ if and only if the Nat\PATL[r]  formula $\coop{A}^{=1}_{|\setpos|}(\F t)$ holds on $\System$. Indeed, memoryless strategies cannot have a complexity higher than $|\setpos|$, the number of states in the MDP, and so available strategies coincide.  
    In Proposition 2 of~\cite{DBLP:conf/aaai/ChatterjeeCD16}, it is shown that finding strategies for POMDPs with memoryless \emph{randomized} strategies and almost-sure reachability objective is \NP-hard. It uses a reduction from Lemma 1 of~\cite{DBLP:conf/hybrid/ChatterjeeKS13}, that only uses deterministic strategies. As such, finding strategies for POMDPs with memoryless \emph{deterministic} strategies and almost-sure reachability objective is \NP-hard, so it is the model checking Nat\PATL[r] with behavioral natural deterministic strategies for the coalition.
\end{proof}

\begin{theorem}\label{prop:mcheck-det-patls}
Model checking Nat\PATLs[r] (respectively Nat\PATLs[R]) with deterministic natural  strategies for the coalition is in \DNEXPTIME.
\end{theorem}

\begin{proof}

    Let $\varphi$ be a Nat\PATLs[R] formula and $\System = ( \setpos, \legal, \trans, \val)$ be a CGS. 
    We guess a polynomial witness consisting of one deterministic strategy $\sigma[\pos,\varphi',\ag]$ with complexity $k$ for every $\pos\in \setpos$,  every subformula $\varphi' = \coop{\coalition}^{\bowtie d}_{k}(\varphi'')$ of $\varphi$, and every $\ag\in\coalition$. We now show how to check in \DEXPTIME that formula $\varphi$ holds given such a witness. 
    
    For every formula $\varphi''$ that does not contain any coalition operator, it can be decided with in \DEXPTIME time if $\pos \models \coop{\Ag\backslash \coalition}^{\conjugate{\bowtie} 1-d}_{k}(\neg \varphi'')$ holds when some strategies $\sigma[\pos,\varphi',\ag]$ are fixed for agents $\ag\in\coalition$. Indeed, after fixing such strategies, we obtain an MDP and an \LTL formula $\neg \varphi''$ that can be model checked in \DEXPTIME~\cite{DBLP:journals/jacm/CourcoubetisY95}.  

    Going through $\varphi$ in a bottom-up manner, we iteratively replace every subformula $\varphi' = \coop{\coalition}^{\bowtie d}_{k}(\varphi'')$ by a new formula that is only true in states $\pos$ where $\coop{\Ag\backslash \coalition}^{\conjugate{\bowtie} 1-d}_{k}(\neg \varphi'')$ does not hold assuming agents in $\coalition$ follow strategies $\sigma[\pos,\varphi',\ag]$. We only need to check a polynomial number of such formulas, and this returns whether $\varphi$ holds or not. 
\end{proof}

\begin{theorem}\label{prop:mcheck-det-patls-h}
Model checking Nat\PATLs[r] (respectively Nat\PATLs[R]) with deterministic natural  strategies for the coalition is \DEXPTIME-hard.
\end{theorem}

\begin{proof}
    We use a reduction from \LTL model checking on MDPs, which is \DEXPTIME-complete.
    Given an \LTL formula $\varphi$, a threshold $d\in [0,1]$ and a CGS $\System$ with only one agent $Ag$, we say $\phi$  holds with at least probability $d$ on $\System$ if and only if the Nat\PATLs[r] formula $\coop{\varnothing}^{\geq 1-d}_{k}(\neg \varphi)$ holds on MDP $\System$: this formula states that for any strategy of the agent (without any complexity bound, since the coalition is empty), formula $\neg\varphi$ holds with probability at least $1-d$: this only happens if there is no strategy ensuring $\varphi$ with at least probability $d$.  
\end{proof}

When considering probabilistic strategies, we follow the same technique as~\cite{aminof2019probabilistic} to  reduce the problem to model checking real arithmetic. Since \LTL is subsumed by Nat\PATLs[R], we also have a doubly-exponential blowup. On the other hand, with Nat\PATL[R], we roughly follow an idea introduced for stochastic game logic~\cite{DBLP:journals/acta/BaierBGK12} to avoid this blowup. As in the proof of Theorem~\ref{prop:mcheck-det-patl}, it is sufficient to consider reachability and invariance problems, both of which are polynomial.  The same holds for Nat\PATLs[r]. 
Next, we consider the model checking of behavioral strategies. 
We remark that the \DEXPTIME-hardness from Theorem~\ref{prop:mcheck-det-patls-h} also applies to Nat\PATL[r] and Nat\PATLs[r] with behavioral natural strategies. We now give model checking algorithms and their complexity. 

\begin{theorem}
\label{thm:mc-sr}
Model checking Nat\PATLs[r] (respectively Nat\PATLs[R]) with behavioral natural strategies for the coalition is in \TEXPSPACE.
\end{theorem}

\begin{proof}[Sketch of proof]
    Probabilistic Strategy Logic (\PSL) with an additional behavioral natural strategies operator $\exists^{nat}_{k}$ captures Nat\PATLs[R], and we show its model checking is in \TEXPSPACE. We give some additional details on \PSL in Definition~\ref{def:PSL-syntax} and~\ref{def:PSL-semantics} in appendix. 
    When translating a Nat\PATLs[R] formula into a \PSL formula in a bottom-up manner, assuming formula $\varphi$ can already be translated into some $PSL(\varphi)$ without any complexity blowup, the $\coop{\coalition}^{\bowtie d}_{k} \varphi$ subformulas, can be translated as $\exists^{nat}_{k}\sigma\forall\mu \mathbb{P}_{\coalition\rightarrow \sigma,\ \Ag\backslash\coalition\rightarrow \mu} (PSL(\varphi))\bowtie d$: a coalition satisfies $\varphi$ iff there exists a natural strategy for the coalition such that for all strategies of the other agents, the \PSL translation of $\varphi$ holds. To model check the operator $\exists^{nat}_{k}$, we modify the proof of Theorem 1 of~\cite{aminof2019probabilistic} showing that model checking \PSL with memoryless strategies is in \TEXPSPACE. This proof translates \PSL into real arithmetic, and a variable $r_{x,s,a}$ represents the probability for strategy $x$ to take action $a$ in state $s$. We can extend this notation to behavioral natural strategies: for a strategy with complexity $k$, we replace variables $r_{x,s,a}$ by 
    $r_{x,s,a,q}$ where $q$ is the current state of the automata representing the regular expressions of a behavioral natural strategy: 
    $$\bigvee_{\text{strategies }\strat,\ compl(\strat)\leq k}\bigwedge_{(\regular, a)\in\strat} \mathcal{A}_{(\regular, a)}$$
    is an automaton with current state $q[\regular]$. 
    We state that two probabilities are equal if they are accepted by the same regular expressions: \begin{multline*}\bigwedge_{(\regular, a)\in\strat}acc(q[\regular],\mathcal{A}_{(\regular, a)}) \wedge acc(q'[\regular],\mathcal{A}_{(\regular, a)}) \\\Rightarrow r_{x,s,a,q} =r_{x,s,a,q'}\end{multline*}
    Both are exponential in the largest $k$ in the formula, since there are exponentially many possible automata of size less or equal to $k$, and we need to describe at most $k$ of them in every conjunction. Nothing else is changed in the proof of~\cite{aminof2019probabilistic}, and thus we have a \TEXPSPACE complexity in the size of the Nat\PATLs[R] formula, exponential in the size of the system and \DEXPSPACE in the largest complexity $k$ used in the formula. 
\end{proof}

\begin{theorem}
Model checking Nat\PATL[r] (respectively Nat\PATL[R]) with behavioral natural strategies for the coalition is in \EXPSPACE.
\end{theorem}

\begin{proof}[Sketch of proof]
    The proof is slightly more complicated than the previous one. We again adapt the proof of Theorem 1 of~\cite{aminof2019probabilistic}. In the proof of Theorem~\ref{thm:mc-sr}, we translate our fragment of \PSL with  natural strategies to real arithmetic. The only exponential blowup comes when translating the coalition operator into $\exists^{nat}_{k}\sigma\forall\mu \mathbb{P}_{\coalition\rightarrow \sigma,\ \Ag\backslash\coalition\rightarrow \mu} (PSL(\varphi))\bowtie d$ where $\varphi$ is assumed to be an \LTL formula whose atoms are either propositional variables, or variables representing other formulas starting with $\mathbb{P}$.   This translation constructs a deterministic Rabin automaton whose size is exponential in the size of the CGS, double exponential in the size of $\psi$, and uses a number of quantifiers double exponential in the size of $\psi$. Since we consider Nat\PATL[R], this LTL formula may only be either $\X \varphi$ or $\varphi \until\varphi'$, where $\varphi$ and $\varphi'$ have been inductively represented as Boolean formulas. Proposition 5.1 and Theorem 5.2 of~\cite{AlurHK02} show that such formulas can be polynomially translated to either reachability or invariance games, which can be done using an automaton and a number of variables both polynomial in the size of the CGS and $\psi$. Since model checking real arithmetic is exponential in the number of quantifiers of the formula~\cite{DBLP:journals/jcss/Ben-OrKR86,fitchas1987algorithmes}, we obtain that model checking Nat\PATL[R] with behavioral natural strategies for the coalition is in \EXPSPACE.
\end{proof}
 
\newcommand{\denotation}[2][]{{\llbracket #2\rrbracket}^{#1}}
\newcommand{\set}[1]{\{#1\}}
\newcommand{\lexpr}{\preceq_e}
\newcommand{\ldist}{\preceq_d}
\newcommand{\Lsys}{\mathcal{L}}
\newcommand{\onlabel}[1]{\colorbox{white}{{{#1}}}}
\newcommand{\Enatstr}[2]{\exists s_{#1}^{\leq #2}}
\newcommand{\Anatstr}[2]{\forall s_{#1}^{\leq #2}}
\newcommand{\Estr}[1]{\exists s_{#1}}
\newcommand{\Astr}[1]{\forall s_{#1}}

\section{Expressivity}\label{sec:expressivity}
We now compare the expressive power of Nat\PATLs\  to that of \PATLs. %
We first recall the notions of distinguishing and expressive powers.

\begin{definition}[Distinguishing power and expressive power \cite{wang2009expressive}]\label{def:dist}
Consider two logical systems $\Lsys_1$ and $\Lsys_2$, with their semantics (denoted $\models_{\Lsys_1}$ and $\models_{\Lsys_2}$, resp.) defined over the same class of models $\mathcal{M}$. 
We say that $\Lsys_1$ is \emph{at least as distinguishing} as $\Lsys_2$ (written: $\Lsys_2\ldist \Lsys_1$) iff for every pair of models $M,M'\in\mathcal{M}$ , 
if there exists a formula $\phi_2\in{L}_2$ such that $M \models_{\Lsys_2}{\phi_2} $ and $M' \not \models_{\Lsys_2}{\phi_2}$, then there is also $\phi_1\in{L}_1$ with $M$ $\models_{\Lsys_1}{\phi_1} $ and $ M' \not \models_{\Lsys_1}{\phi_1}$.

Moreover, $\Lsys_1$ is \emph{at least as expressive} as $\Lsys_2$ (written: $\Lsys_2\lexpr \Lsys_1$) iff for every $\phi_2\in{L}_2$ there exists $\phi_1\in{L}_1$ such that, for every  $M\in\mathcal{M}$, we have $M\models_{\Lsys_2} {\phi_2} $ iff $M\models_{\Lsys_1} {\phi_1}$.
\end{definition}

Nat\PATLs  and \PATLs are based on different notions of strategic ability. 
As for the deterministic setting with \ATL, each behavioral natural strategy can be translated to a behavioral combinatorial one (i.e., mappings from sequences of states to actions), but not vice versa. Consequently, \PATLs can express that a given coalition has a combinatorial strategy to achieve their goal, which is not expressible in Nat\PATLs. On the other hand, Nat\ATLs allows expressing that a winning natural strategy with bounded complexity does not exist, which cannot be captured in \PATLs. 
Now we show that Nat\PATLs allows expressing properties %
that cannot be captured in \PATLs, and vice versa.

\begin{theorem}\label{prop:distpower}
For both memoryless and recall semantics:
\begin{itemize}
    \item \textnormal{Nat}\PATL (resp. \textnormal{Nat}\PATLs) and $\PATL$ (resp, \PATLs) have incomparable \emph{distinguishing} power over  \CGS. %
    \item \textnormal{Nat}\PATL (resp. \textnormal{Nat}\PATLs) and $\PATL$ (resp, \PATLs) have incomparable \emph{expressive} power over \CGS. %
\end{itemize}

\end{theorem}

\begin{proof}[Sketch of proof]
The proof can be obtained by a slight adjustment of the proofs  regarding the %
distinguishing power of Quantified \SL with Natural Strategies (Nat\SLF) and  Quantified \SL (\SLF) with combinatorial strategies (Propositions 8 and 9 of 
\cite{DBLP:conf/atal/BelardinelliJMM22}). Since the logics considered there are not Boolean, notice that 
(i) the satisfaction value $1$ and $-1$ represent whether a formula is satisfied by a weighted CGS or is not; (ii)   the counterexamples considered in the proofs are  weighted deterministic CGS where the value of atomic propositions are restricted to $-1$ and $1$, and thus can be easily transformed in Boolean deterministic CGSs; and (iii) the counterexamples showing that the logics have incomparable distinguishing power are easily converted to \PATL and Nat\PATL formulas by changing the deterministic coalition operators into probabilistic coalition operators with probability $\geq1$.

Thus, we have that, for both memoryless and recall semantics:
\begin{itemize}
    \item $\textnormal{Nat}\PATL$ (resp, $\textnormal{Nat}\PATLs$) $\not\ldist \PATL $ (resp, \PATLs) %
    \item $\PATL$ (resp, $\PATLs$) $\not\ldist \textnormal{Nat}\PATL $ (resp, $\textnormal{Nat}\PATLs$) 
\end{itemize} 

That is,  $\textnormal{Nat}\PATL$ (resp, $\textnormal{Nat}\PATLs$)  and $\PATL$ (resp, $\PATLs$) have incomparable distinguishing power.

From the definitions of distinguishing power and expressive power, it is easy to see that, for two logical systems $\Lsys_1$ and $\Lsys_2$, $\Lsys_2\lexpr \Lsys_1$ implies $\Lsys_2\ldist \Lsys_1$. By transposition, we also get that $\Lsys_2\not\ldist \Lsys_1$ implies $\Lsys_2\not\lexpr \Lsys_1$. 
Thus,    $\textnormal{Nat}\PATL$ (resp, $\textnormal{Nat}\PATLs$)  and $\PATL$ (resp, $\PATLs$) have incomparable expressive power.
\end{proof}

\section{Conclusion}\label{sec:ccl}
In this work, we have defined multiple variations of \PATL with natural strategies, and studied their model-checking complexity. We have illustrated with multiple examples the relevance of the probabilistic setting, which can represent uncertainty in a very precise way, and the interest in natural strategies, that are both efficient and much closer to what a real-world agent is expected to manipulate. 

In terms of model checking, the \NP-completeness of Nat\PATL with deterministic strategies is promising, and shows we can capture POMDPs with bounded memory without any significant loss. While the \DNEXPTIME complexity for Nat\PATLs with deterministic strategies is high, we have shown a close lower bound, namely \DEXPTIME-hardness. With probabilistic strategies, the \EXPSPACE membership of Nat\PATL is quite similar to the result of~\cite{DBLP:journals/acta/BaierBGK12}, and the \TEXPSPACE membership of Nat\PATLs is also similar to~\cite{aminof2019probabilistic}. Since this exponential space blowup comes from the use of real arithmetic to encode probabilities, any improvement would likely come from the introduction of a totally new technique. Similarly, the doubly-exponential blowup between \PATL and \PATLs comes from the \DEXPTIME-completeness of \LTL model checking on MDPs. We also keep the \DEXPTIME-hardness from the deterministic case. To our knowledge, similar works~\cite{DBLP:journals/acta/BaierBGK12,aminof2019probabilistic}, do also not give different lower bounds between deterministic and probabilistic strategies. A possible approach would be to use a construction from POMDPs, more precisely either~\cite{DBLP:conf/uai/Junges0WQWK018}, showing
that synthesis on POMDPs with reachability objectives and bounded memory is \NP-complete for deterministic strategies and \ETR-complete
for probabilistic finite-memory strategies or~\cite{oliehoek2012decentralized}, showing that finding a policy maximizing a reward on a decentralized POMDPs with full memory is \NEXPTIME-complete). 
Our results on expressivity mean that there are properties of stochastic MAS with natural strategies that cannot be equivalently translated to properties based on combinatorial strategies, and vice versa.  

The proof of Theorem~\ref{thm:mc-sr} shows that we could extend natural strategies to \PSL, but it would be difficult to get a better result than our \TEXPSPACE complexity. Considering qualitative \PATLs or \PSL (\ie\ only thresholds $>0$ and $=1$) may yield a better complexity. For the quantitative setting, i.e., thresholds such as $>\frac{1}{2}$, techniques from the field of probabilistic model checking can be applied, e.g., graph analysis, bisimilation minimization, symbolic techniques, and partial-order reduction~\cite{katoen2016probabilistic}.  Another direction would be to consider \emph{epistemic operators}. Indeed, many applications involving agents with a reasonable way to strategize also have to take into account the knowledge and beliefs of these agents. As such, we would have to find a good epistemic framework such that natural strategies keep the desired balance between expressivity and complexity.
 
\section{Acknowledgements}
This project has received funding from the European Union’s Horizon 2020 research and innovation programme under the Marie Skłodowska Curie grant agreement No 101105549. This research has been supported by the PRIN project RIPER (No. 20203FFYLK), the PNRR MUR project PE0000013-FAIR, the InDAM 2023 project “Strategic Reasoning
in Mechanism Design”, and the DFG Project POMPOM (KA 1462/6-1).

%\clearpage
\bibliography{ref} 

\clearpage
\appendix
\begin{center}
{\Large Appendix}
\end{center}
\section{Probabilistic Strategy Logic}
\label{app:PSL}
Given a set of atomic proposition $\APf$, agents $\Ag$, and strategy variables $\Vvar$, we introduce Probabilistic Strategy Logic (\PSL), defined as in~\cite{aminof2019probabilistic}.

\begin{definition}[\PSL\ syntax]\label{def:PSL-syntax}
The syntax of \PSL  is defined by the grammar: 
\begin{align*}
	\varphi  &::= p \mid  {\varphi \lor  \varphi} \mid \neg \varphi \mid \exists x. \varphi \mid \tau \leq \tau
 \\
        \tau &::= c \mid \tau^{-1} \mid \tau - \tau \mid \tau + \tau \mid \tau \times \tau \mid \mathbb{P}_{\beta}(\psi)
        \\
        \psi  &::=\varphi \mid \neg \psi \mid {\psi \lor \psi} \mid \X \psi \mid \psi \until \psi
\end{align*}
\end{definition}

Formulas $\varphi$ are called history formulas, formulas $\tau$ are called arithmetic terms, and formula $\psi$ are called path formulas. In the proof of Theorem~\ref{thm:mc-sr}, we introduce a new history formula: $\varphi ::= \exists^{nat}_{k} x. \varphi$, stating that there exists a natural strategy with complexity $k$. We remark that apart from Boolean formulas over propositional variables, the only history formulas and arithmetic terms we use in this proof to encode Nat\ATLs are of the form $\exists^{nat}_{k}\sigma.\ \forall\mu.\  \mathbb{P}_{\coalition\rightarrow \sigma,\ \Ag\backslash\coalition\rightarrow \mu} (PSL(\varphi))\bowtie d$.

\begin{definition}[\PSL semantics]
\label{def:PSL-semantics}
We denote by $\Sigma$ the set of all strategies (natural or not) and use $\mu^{\profile{\strat_\coalition}}_\pos$ to range over the measures in $out_\coalition(\profile{\strat_\coalition},\pos)$. \PSL formulas are interpreted in a stochastic CGS $\System$, 
a valuation $\nu:\ \Vvar\to\Sigma$  
  and a path  $\iplay$. The semantics of history formulas is as follows:
\begingroup
\allowdisplaybreaks
\begin{align*}
 \System,\nu,\iplay &\models p & \text{ iff } & p \in \val(\iplay_0)\\
 \System,\nu,\iplay &\models \neg \varphi & \text{ iff } & \System,\nu,\iplay \not \models \varphi \\
 \System,\nu,\iplay &\models \varphi_1 \lor \varphi_2 & \text{ iff }&  
 \System,\nu,\iplay \models \varphi_1  \text{ or } \System,\nu,\iplay \models \varphi_2
 \\
 \System,\nu,\iplay &\models \exists x. \varphi & \text{ iff }& 
 \exists\sigma\in\Sigma.\ \System,\nu[x\mapsto \sigma],\iplay \models \varphi \\
 \System,\nu,\iplay &\models \tau_1 \leq \tau_2 & \text{ iff }& 
 val{\nu,\iplay}(\tau_1)\leq val_{\nu,\iplay}(\tau_2) \\
\end{align*} 
\endgroup
where
\begin{align*}
 &val{\nu,\iplay}(c) = c \text{ and } val{\nu,\iplay}(\tau^{-1}) = val{\nu,\iplay}(\tau))^{-1} \\
 &val{\nu,\iplay}(\tau \oplus \tau') =  val{\nu,\iplay}(\tau)  \oplus val{\nu,\iplay}(\tau') \text{ for } \oplus\in \{-,+,\times\} \\
 &val{\nu,\iplay}(\mathbb{P}_{\beta}(\psi)) = \mu^{\profile{\nu\circ\beta}}_{\iplay_0}(\{\pi\ :\ \System,\nu,\iplay\models\psi\})
\end{align*} 
The semantics of path formulas is as follows:
\begingroup
\allowdisplaybreaks
\begin{align*}
 \System,\nu,\iplay &\models \neg\varphi & \text{ iff } & \System,\nu,\iplay \not\models \varphi\\ 
 \System,\nu,\iplay &\models \psi_1 \lor \psi_2 & \text{ iff }&  
 \System,\nu,\iplay \models \psi_1  \text{ or } \System,\nu,\iplay \models \psi_2\\
 \System,\nu,\iplay &\models \X \psi & \text{ iff } & \System,\nu,\iplay_{\geq 1} \models \psi \\
\System, \iplay  & \models \psi_1 \until \psi_2 & \text{ iff } &  \exists k \geq 0 \text{ s.t. } \System,\nu,\iplay_{\geq k} \models \psi_2 \text{ and } 
\\ & & &
\forall j \in [0,k).\, \,  \System,\nu,\iplay_{\geq j}\models \psi_1
\end{align*} 
\endgroup

We also define the semantics of the natural strategy existential quantifier:
\begin{align*}
 \System,\nu,\iplay &\models \varphi_1 \lor \varphi_2 & \text{ iff }& 
 \exists\sigma\in\Sigma,\ \sigma \text{ natural },\ \System,\nu[x\mapsto \sigma],\iplay \models \varphi \\
\end{align*} 

\end{definition}

\end{document}